\documentclass[%
reprint,
superscriptaddress,
 amsmath,amssymb,
 aps,
pra,
]{revtex4-2}

\usepackage{graphicx,color}% Include figure files
\usepackage{dcolumn}% Align table columns on decimal point
\usepackage{bm}% bold math
\usepackage{graphicx}
\usepackage{textcomp}
\usepackage{amsthm}
\usepackage{mathtools}
\usepackage{algorithm, algpseudocode}
\usepackage{enumitem}
\usepackage{physics}
\usepackage[export]{adjustbox}
\usepackage{appendix}
\usepackage[normalem]{ulem}
\mathtoolsset{showonlyrefs=true}

\algrenewcommand\algorithmicindent{1em}

\theoremstyle{remark}

\newtheorem{theorem}{Theorem}[]

\newtheorem{proposition}{Proposition}[]
\newtheorem{lemma}{Lemma}[]

\usepackage[dvipsnames]{xcolor}

\newcommand{\EX}{\mathbf{E}}
\newcommand{\PR}{\mathbf{P}}

\newcommand{\I}{\mathbb{I}}
\newcommand{\cov}{\text{cov}}
\newcommand{\vari}{\text{var}}

\def\mc{\mathcal}
\def\mbf{\mathbf}

\usepackage{xr}
\makeatletter
\newcommand*{\addFileDependency}[1]{% argument=file name and extension
\typeout{(supplementary)}% latexmk will find this if $recorder=0
\@addtofilelist{supplementary}
\IfFileExists{supplementary}{}{\typeout{No file supplementary.}}
}\makeatother

\newcommand*{\myexternaldocument}[1]{%
\externaldocument{supplementary}%
\addFileDependency{supplementary.tex}%
\addFileDependency{supplementary.aux}%
}
\myexternaldocument{supplementary}

\usepackage{xstring}
\newcommand{\secref}[1]{\IfBeginWith{#1}{sec:}{S-}{}\ref{#1}}

\begin{document}

\preprint{APS/123-QED}

\title{\textbf{Effects of Correlated Errors on Quantum Memory} 
}% 

\author{Smita Bagewadi}
\email{Contact author: ee19d002@smail.iitm.ac.in}
\affiliation{%
 Department of Electrical Engineering, Indian Institute of Technology Madras,  India
  %Include all institutions where the work was conducted: department or division, institution, city, state (if relevant), and country, in this order.
}
\author{Avhishek Chatterjee}%
 \email{Contact author: avhishek@ee.iitm.ac.in}
\affiliation{%
 Department of Electrical Engineering, Indian Institute of Technology Madras, India
  %Include all institutions where the work was conducted: department or division, institution, city, state (if relevant), and country, in this order.
}%

%\affiliation{
 %First affiliation for this author
%}%
%\affiliation{
% second institution for this author
%\affiliation{%


\date{\today}% It is always \today, today,

\begin{abstract}
Recent results on constant overhead LDPC code-based fault-tolerance against i.i.d. errors naturally lead to the question of fault-tolerance against errors with long-range correlations. Ideally, any correlation can be captured by a joint (system and bath) Hamiltonian. However, an arbitrary joint Hamiltonian is often intractable, and hence, the joint Hamiltonian model with pairwise terms was introduced and developed in a seminal series of papers \cite{TerhalB2005,AliferisGP2005,AharonovKP2006}. However, the analysis of the new constant overhead codes in that error model appears to be quite challenging.

In this paper, to model correlated errors in quantum memory, we introduce a  correlation model which is a generalization of the well-known hidden random fields. This proposed model, which includes stationary and ergodic (non-Markov) error distributions, is shown to capture correlations not captured by the joint Hamiltonian model with pairwise terms. On the other hand, building on non-i.i.d. measure concentration, we show that for a broad class of non-Markov and (possibly) non-stationary error distributions, quantum Tanner codes ensure an exponential retention time (in the number of physical qubits), when the error rate is below a threshold. An implication of these results is that the rate of decay of the correlation with distance does not necessarily differentiate between good and bad correlation.

%Towards that, as an intermediate result, we show a quadratic (in physical qubits) upper bound on the retention time for a $1$D correlation within our model. This $1$D correlation is also shown to have an arbitrarily low error rate and quadratically dying correlation. These, in turn, imply that the concerned $1$D correlation cannot be captured by the joint Hamiltonian model with pairwise terms. 

%This shows that the proposed model is tractable enough for studying the constant overhead codes, thanks to the rich probability literature around it.
\end{abstract}

\maketitle

%\tableofcontents

Quantum LDPC codes can provably offer constant overhead fault-tolerance for quantum memories in the presence of i.i.d. and local stochastic errors \cite{gottesman2014fault, fawzi2020constant, BravyiCG2024high}. Recent developments in quantum codes gradually led  to constant overhead codes with linear distance \cite{PanteleevK2021,LeverrierZ2022,LeverrierZ2023,GuPT2023}. This has made a compelling case for the use of constant overhead codes in fault-tolerance. However, the fault-tolerance offered by such codes against correlated errors, which is a pressing practical issue \cite{WilenAKS2021}, is yet to be understood. 

Stabilizer codes are the most widely used quantum codes for error correction. In the stabilizer framework, errors are effectively quantized to Pauli errors, and their locations are inferred from error syndromes. Thus, in the stabilizer framework, a model for correlated errors is essentially a model of correlation between error locations. Local stochastic error model \cite{aharonov1997fault,AharonovB2008} is a well-known  model for correlation between error locations  and is the widely used error model for studying fault-tolerance using constant overhead codes \cite{gottesman2014fault, fawzi2020constant}. However, the local stochastic noise model is restricted by exponential correlation decay \cite{TerhalB2005}, which is not practical.

% Arbitrarily correlated quantum noise can be modeled by an appropriate choice of a joint (system and bath)  Hamiltonian. However, such arbitrary Hamiltonians are often analytically intractable. {\color{blue} Two approaches were developed for modeling ....from stabilizer point of view every error is quantized to Pauli..so model the error locations...another is a tractable hamiltonian . . .}

In \cite{TerhalB2005,AliferisGP2005,AharonovKP2006}, a much broader model of correlation was developed using a pairwise joint (system and bath) Hamiltonian approach. In this model, when the correlation between two qubit locations dies as inverse of their distance (power dimension), concatenated codes were shown to achieve fault-tolerance using poly-logarithmic overhead. Although this model is broader than the local stochastic model, it is not tractable for constant overhead codes.

In this paper, we propose a model for correlated errors based on hidden random fields, which has been a topic of interest in probability theory and statistical physics \cite{KunschGK1995,CamiaJN2018, Chatterjee2018}. Our proposed model, which includes all stationary and ergodic errors within its expanse, is proven to capture correlations beyond the pairwise Hamiltonian model. As an intermediate result, we derive a quadratic upper bound on the retention time of any error-corrected quantum memory under a class of correlated errors. On the other hand, building on non-i.i.d. concentration results \cite{KontorovichR2008}, we show that for a broad class of correlated errors, quantum Tanner codes can keep a memory corrected for an exponential time, if the error rate is below a threshold.

The proposed correlated error model is a broad generalization of the local stochastic error model  \cite{aharonov1997fault,AharonovB2008}. More specifically, it is a (possibly) non-stationary and (possibly) non-Markov generalization of Eqn. 2.3 and 2.5 in \cite{AharonovB2008}. Due to the presence of a rich literature on random fields \cite{KunschGK1995,CamiaJN2018,Chatterjee2018} and non-i.i.d. measure concentration \cite{BoucheronLM2013,KontorovichR2008}, this model is likely to be analytically tractable for studying future fault-tolerance techniques. %It is worth mentioning that i.i.d. measure concentration  played a pivotal role in modern theoretical machine learning.  

Quantum Tanner codes are a rich class of constant overhead codes with very good minimum-distance properties. In \cite{PanteleevK2021}, it was shown that quantum LDPC codes, a special class of Tanner codes, have an  almost linear minimum distance. Later an exact linear minimum distance result for a class of quantum Tanner codes  was derived in \cite{LeverrierZ2022}. We refer to \cite{LeverrierZ2022} for a detailed discussion of the quantum Tanner codes. In this paper, we only use the linear minimum-distance property of the quantum Tanner codes from \cite{LeverrierZ2022} and do not invoke any other properties of them.

\subsection{{\bf Hidden Random Field Model for Errors}}
\label{sec:model}
A $k$ qubit state is encoded into an $n$ qubit state ($n>k$) using a quantum code and is placed in a memory at time $t=0$. This memory is periodically corrected for at discrete time instants $t \in \{1, 2, \ldots\}$. Between $t$~th and $(t+1)$~th correction epochs, the environment causes each qubit to decohere and become erroneous. We model error in a qubit by the generic operator $\mc{E}$, which is representative of the standard Pauli operators in the stabilizer framework. 

Between any two correction epochs $t$ and $t+1$, an error occurs at a qubit location with probability $\epsilon$, i.e., with probability $\epsilon$ the qubit is operated by $\mc{E}$ and with probability $1-\epsilon$ it is operated by the identity operator $I$. However, the events that correspond to the occurrences of errors are not independent of each other. 

We represent the occurrence of an error at qubit $i$ between epochs $t$ and $t+1$ by a $\{0,1\}$-valued random variable $Y_{i,t}$, where $Y_{i,t}=1$ represents error.
Thus, if the state of qubit $i$ at the end of correction epoch $t$ is $\rho_{i,t}$, its state right before the correction epoch $t+1$ is given by $\mc{N}_{i,t}(\rho_{i,t})$. Here, $\mc{N}_{i,t}=\EX_{Y_{i,t}\sim \text{Bernoulli}(\epsilon)}\left[(1-Y_{i,t}) I + Y_{i,t} \mc{E}\right]$, which is exactly the same as Eqn. 2.5 in \cite{AharonovB2008}. However, since the errors across the qubits are not independent, the joint evolution of $n$ qubit state $\rho^{(n)}_{t}$ is given by $\mc{N}^{(n)}_t(\rho^{(n)}_{t})$, where 
\begin{equation}
\mc{N}^{(n)}_t=\EX_{\{Y_{i,t}: 1 \le i \le n\}} \otimes_{i=1}^n \left[(1-Y_{i,t}) I + Y_{i,t} \mc{E}\right]. \nonumber
\end{equation}
In general, an arbitrary joint Bernoulli distribution of $\mbf{Y}_t:=\{Y_{i,t}: 1 \le i \le n\}$ can capture any  correlated error pattern. Due to analytical tractability, the distribution of $\mbf{Y}_t$ was restricted in \cite{AharonovB2008}. In this work, we consider a hidden random field model for $\mbf{Y}_t$: the distribution of $\mbf{Y}_t$, $\mbf{p}(\mbf{Y}_t)=\sum_{\mbf{X}_t} \mbf{p}(\mbf{X}_t) \mbf{q}(\mbf{Y_t}|\mbf{X}_t)$, where $\mbf{X}_t:=\{X_{i,t}: 1 \le i \le n \}$ are random variables taking values in a countable space $\mc{X}$ (possibly infinite). $\mbf{p}$ and $\mbf{q}$ are, respectively, an arbitrary joint distribution of $\mbf{X}$ and an arbitrary conditional distribution of $\mbf{Y}$ given $\mbf{X}$.

This error model is inspired by the fact that a non-Markov subclass of this model, called the hidden Markov random field model, is known to be dense among the class of all stationary and ergodic joint Bernoulli distributions \cite{KunschGK1995}. Note that errors from any time-invariant joint Hamiltonian would be limited to stationary and ergodic errors.  The proposed model has a natural physical interpretation: variables $\mbf{X}_t$ model physical parameters within the memory such as local temperature, local magnetic field, and local value of external radiation at time $t$. Its spatial distribution is modeled by $\mbf{p}(\mbf{X}_t)$. These physical parameters, in turn, influence the manifestation of errors $\mbf{Y}_t$ through the conditional distribution $\mbf{q}(\mbf{Y}_t|\mbf{X}_t)$, which depends on the choice of technology. 

In this paper, we restrict ourselves to the following class of error distributions, which we call the generalized hidden Markov random field model: $\mbf{X}$ is a Markov random field (MRF) and $\mbf{q}(\mbf{Y}_t|\mbf{X}_t)=\prod_{i=1}^n q_i(Y_{i,t}|\mbf{X}_t)$. In particular, we restrict to $1$D arrangment of qubits and hence, to $1$D MRF.
An MRF is a generalization of the well-known Ising model for spin glass systems. For a $1$D MRF $\mbf{X}$: $\mbf{p}(\mbf{X})=p(X_1)\prod_{i=1}^{n-1} p_i(X_{i+1}|X_i)$. 
$\mbf{Y}_t$ is called a hidden MRF if $\mbf{X}_t$ is an MRF and $\mbf{q}(\mbf{Y}|\mbf{X})=\prod_{i=1}^n q_i(Y_i|X_i)$. The above class of distribution is a generalization of it, and hence the name. Later, we show that this generalized hidden MRF model is strictly broader than the pairwise joint Hamiltonian model.

% to the class of conditional distributions $\mbf{q}$ from $\mc{X}^n$ to $\{0,1\}^n$ satisfying $\mbf{q}(\mbf{Y}_t|\mbf{X}_t)=\prod_{i=1}^n q_i(Y_{i,t}|\mbf{X}_t)$. Here, $q_i$ are arbitrary conditional distributions from $\mc{X}^n$ to $\{0,1\}$. Later, we show that this restricted class of distributions is broader than the pairwise joint Hamiltonian model.

It is clear that $\{\mc{N}_{i,t}\}$ can be obtained by taking partial trace of $\mc{N}^{(n)}_t$. However, in general, $\mc{N}^{(n)}_t$ is not equal to the tensor product of the single-qubit operators, $\{\mc{N}_{i,t}\}$. We assume that the process $\{Y_{i,t}: 1 \le i \le n\}$ is stationary with respect to time $t$, and possibly non-Markov. %The joint distributions of errors are independent across time, i.e., $\{Y_{i,t}: 1 \le i \le n\}$ are independent of $\{Y_{i,t'}: 1 \le i \le n\}$ for any $t\neq t'$ and are identically distributed. 

For simplicity, we would drop $t$ from the notations  $Y_{i,t}$, $\mbf{Y}_t$, $X_{i,t}$, $\mbf{X}_t$, $\mc{N}_{i,t}$ and $\mc{N}_t$, whenever the context is clear. The error rate for qubit $i$ is $\EX[Y_{i}]$ and the error rate of the memory is $\frac{1}{n} \sum_{i=1}^n \EX[Y_{i}]$.

The state $\mc{N}^{(n)}_t(\rho^{(n)}_{t})$ is corrected using a fault-tolerant quantum circuit. After correction,  the state in the memory at epoch $t+1$ becomes $\rho^{(n)}_{t+1}$, which is again corrupted by correlated errors in the memory and is again corrected at epoch $t+2$ by the correcting quantum circuit. 

We say that a memory with a particular fault-tolerant scheme has a retention time or lifetime $T$, if for some decoder $\mc{D}$ and for all $t\le T$, $\mc{D}\left(\mc{N}^{(n)}_t(\rho^{(n)}_{t})\right)$ is within a prescribed distance from the original $k$-qubit state. 
We say that a memory has an exponential retention or lifetime if $T \ge \exp(a n)$ for some $a>0$. On the other hand, we say that a memory has at most a linear or a quadratic retention time if $T \le c n$  or $T \le c n^2$, respectively,  for some $c>0$. 

Clearly, exponential retention time is the best achievable retention time since even fully error-corrected classical memory cannot store classical bits against simple erasures any longer than that. On the other hand, a memory with linear or quadratic lifetime is not suitable for storing quantum states for future processing when they are generated by costly quantum computations or simulations. Though our model applies to qubit arrangement in any dimension, in this paper, we consider only $1$D arrangements.

\subsection{{\bf Expanse of the Model}}
\label{sec:converse}

The hidden MRF model, which is a special case of the generalized hidden $1$D MRF model proposed in this paper, is known to be dense among the class of all stationary and ergodic error distributions \cite{KunschGK1995}. Under a time-invariant (system and bath) joint Hamiltonian, errors in the stabilizer framework would be stationary and ergodic. Thus, the proposed error model includes all error distributions in the stabilizer framework that can potentially be generated by a time-invariant (system and bath) joint Hamiltonian. Thus, the proposed model for long-range correlated error is at least as broad as the pairwise joint Hamiltonian model. Next, we show an example of a correlated error distribution from the proposed model that is strictly outside the scope of the pairwise joint Hamiltonian model.

% As mentioned before, a special case of the proposed model, the hidden Markov random field model, can capture any stationary and ergodic error distribution \cite{KunschGK1995}. Under any time-invariant (system and bath) joint Hamiltonian, errors in the stabilizer framework would be stationary and ergodic. Thus, it may be argued that the proposed error model includes all error distributions in the stabilizer framework that can are generated by a time-invariant (system and bath) joint Hamiltonian, not restricted to pairwise joint Hamiltonian. 

Consider the following special case of our model, where the errors $\mathbf{Y}_t$ are i.i.d. in time and have the following distribution at any time $t$. The hidden variables $\mbf{X}$ are i.i.d (in space) Bernoulli and $\mbf{Y}$ are conditionally independent given $\mbf{X}$: $\PR(\mbf{Y})=\sum_{\mbf{x}} \mbf{p}(\mbf{X}) \mbf{q}(\mbf{Y}|\mbf{X})$, and $\mbf{q}(\mbf{Y}|\mbf{X}) = \prod_{i=1}^n q_i(Y_i|\mbf{X})$, where $q_i(Y_i=1|\mbf{X})=1$ if $\sum_i X_i \ge \sum_i \EX[X_i] + C_n \sqrt{n}$ and $q_i(Y_i=1|\mbf{X})=0$ else. Clearly, the errors $\mbf{Y}$ have a stationary distribution.

Under the above model let the channel be $\mc{N}^{(n)}$ and let $\mc{N}_{ij}$ be the partial trace of $\mc{N}$ over all physical qubits but $i$ and $j$. Ideally, if the errors were independent, $\mc{N}_{ij}=\mc{N}_i \otimes \mc{N}_j$. Thus, the diamond distance $||\mc{N}_{ij}-\mc{N}_i \otimes \mc{N}_j||_{\diamond}$ is a good measure of correlation between the noise operators at two different qubit locations. By analyzing the above error model we observe that for this $1$D error model, the correlation (between $Y_i$ and $Y_j$) and the diamond distance decay faster than $1/\mbox{distance}^2$. On the other hand, we also observe from the analysis that the probability of $\{Y_{i,t}=1, 1 \le i \le n\}$ any $t$, i.e., all qubits are in error, is lower bounded by $\frac{c'}{n^2}$, for some $c'>0$. This implies that after $\frac{n^2}{c'}$ correction epochs, with probability $1-\frac{1}{e}$, all qubits in the memory would become erroneous. Thus, no code would be able to retain the memory for more than $\frac{n^2}{c'}$ time.  The following theorem states this result formally, whose detailed proof is presented in Appendix~\ref{sec:prop:badCorr}.

\begin{proposition}
\label{prop:badCorr}
For $X_i$ i.i.d. Bernoulli $\epsilon$, there exists a choice $C_n$ such that the covariance between  two errors $Y_i$ and $Y_j$ and  $||\mc{N}_{ij}-\mc{N}_i \otimes \mc{N}_j||_{\diamond}$ are less than $\frac{c'}{n^2}$ for some constant $c'>0$, and $|\EX[Y_i]-\epsilon| \to 0$ as $n\to \infty$. However, for any error rate $\epsilon>0$ and any fault-tolerant memory with an error-free circuit for correction, the retention time is bounded above by $\alpha n^{2}$ for some constant $\alpha$ independent of $n$, for any tolerable distance (from the initial state) $\le 1-\frac{1}{e}$. $\blacksquare$
\end{proposition} 

The decay of correlation  faster than $1/\text{distance}^2$ is equivalent to a bound of $1/\text{distance}^2$ on the norms of pairwise terms in the joint Hamiltonian model. When the error rate is sufficiently low and pairwise Hamiltonian terms decay faster than $1/\text{distance}^2$, errors are provably correctable for arbitrary time \cite{AharonovKP2006}. Thus, by {\em reductio ad absurdum}, Proposition~\ref{prop:badCorr}  implies that the particular error distribution considered therein, which is a special case of the proposed model, is not captured by the joint Hamiltonian model with pairwise terms.  %interactions among multiple qubits.

\subsection{{\bf Tanner Codes and Long-range Correlations}}
\label{sec:tanner}
% Since hidden Markov random field can capture any stationary and ergodic distribution, next we analyze the performance of Tanner codes under a further generalization of that error model. This generalization would allow us to study non-stationary error distributions. 

We observed above that the hidden MRF model is strictly broader than the best existing error model for long-range correlations. Our next result concerns   error distributions from the hidden MRF model that satisfy a broad class of regularity conditions. We prove that under such correlated errors an appropriately chosen quantum Tanner code can keep a quantum memory corrected for an exponential (in $n$) time, if the error rate is below a threshold. Intuitively, as discussed later, the regularity conditions are related to the rate of decay of correlation between two locations with their geometric distance and the smoothness of the conditional probabilities $\{q_i\}$. 

% A Markov random field (MRF) $\mbf{X}$ is a generalization of the well-known Ising model for spin glass systems, where $\mbf{p}(\mbf{X})=p(X_1)\prod_{i=1}^{n-1} p_i(X_{i+1}|X_i)$. For a hidden MRF: $\mbf{q}(\mbf{Y}|\mbf{X})=\prod_{i=1}^n q_i(Y_i|X_i)$. Here, we consider a generalized hidden MRF with $\mbf{q}(\mbf{Y}|\mbf{X})=\prod_{i=1}^n q_i(Y_i|\mbf{X})$. Clearly, $\mbf{Y}$ is non-Markov and can be non-stationary as well if the $\{q_i\}$ are different. For this error model, as stated below,  a fault-tolerance scheme based on quantum Tanner codes can have  $\exp(\Theta(n))$ retention time when error rate is below a threshold and $\{q_i\}$ follow some technical conditions, which are elaborated later. 

\begin{proposition}
\label{prop:hidden1DMRF}
A quantum memory error corrected using a well-chosen (quantum Tanner) code with constant overhead has a lifetime $\exp(a n)$ for some $a>0$ and $n$ sufficiently large, if the error rate is below a threshold and $\mbf{Y}$ is from a generalized hidden $1$D MRF which satisfies the following conditions: 

For $\theta_i:=\frac{1}{2} \sup_{x,x'} \sum_{a \in \mc{X}}|p_i(a|x) - p_i(a|x')|$, $1 + \max_{1\le i \le n} \sum_{k=i}^n \prod_{j=i}^k \theta_j$ is a constant independent of $n$ and  $\psi(\mbf{X}):=\sum_{i=1}^n q_i(Y_i=1|\mbf{X})=\sum_{i=1}^n \EX[Y_i|\mathbf{X}]$ is $c$-Lipschitz \cite{KontorovichR2008} with respect to the Hamming distance on $\mc{X}$. $\blacksquare$

% \begin{enumerate}
% %\item $0 < \inf_i \inf_{x \in \mc{X}} q_i(0|x) \le \sup_i \sup_{x \in \mc{X}} q_i(0|x)< 1$
% % \item $\mc{X}$ is finite,  $p_i=p$  is ergodic for all $i$  and for some finite positive integer $\tau>1$, $\sup_{x,x'} \sum_{a \in \mc{X}}|\PR(X_{\tau}=a|X_1=x) - \PR(X_{\tau}=a|X_1=x')|<\frac{1}{4}$. %proof using Chung, Lam, Liu and Mitzenmacher Theorem 3.1
% \item Let $\theta_i=\frac{1}{2} \sup_{x,x'} \sum_{a \in \mc{X}}|p_i(a|x) - p_i(a|x')|$ and $1 + \max_{1\le i \le n} \sum_{k=i}^n \prod_{j=i}^k \theta_j$ is a constant independent of $n$. %proof using Kontorovich and Ramanan Theorem 1.2
% \end{enumerate}
\end{proposition}

The condition on $1 + \max_{1\le i \le n} \sum_{k=i}^n \prod_{j=i}^k \theta_j$  is related to $1/\mbox{distance}$  decay of correlation in $1$D. This becomes clear by considering the special case: $\mbf{Y}=\mbf{X}$ and $\mc{X}=\{0,1\}$. Here, $\prod_{j=i}^k \theta_j$ is a bound on $|\EX[Y_k|Y_i=1]-\EX[Y_k|Y_i=0]|$, which is the standard measure of absolute value of correlation in an MRF \cite{CamiaJN2018,Chatterjee2018}. Thus, the condition, $1 + \max_{1\le i \le n} \sum_{k=i}^n \prod_{j=i}^k \theta_j=O(1)$, can be interpreted as the sum of the absolute values of correlations between pairs of locations are $O(1)$. This, in turn, means that in Prop.~\ref{prop:hidden1DMRF}, the absolute value of correlations between any two error variables, $Y_i$ and $Y_k$,  decays faster than $\frac{1}{|k-i|}$. Interestingly, this condition echoes the condition in \cite{AharonovKP2006} for fault-tolerance against long-range correlations.  

Given a discrete space $\mc{X}$, a function $\phi$ from $\mc{X}^n$ to real numbers is called $c$-Lipschitz with respect to Hamming distance if for any two sequences $\mbf{x}$ and $\mbf{x}'$ in $\mc{X}^n$,  $|\phi(\mbf{x})-\phi(\mbf{x}')|$ is upper bounded by $c$ times the number of coordinates where $\mbf{x}$ and $\mbf{x}'$ differ. Thus, the $c$-Lipschitz property is similar to smoothness of a function in an Euclidean space, i.e., a function with bounded derivatives. Thus, intuitively, the $c$-Lipschitz condition on $\sum_{i=1}^n q_i(Y_i=1|\mbf{X})=\sum_{i=1}^n \EX[Y_i|\mathbf{X}]$ which is a function from $\mc{X}^n$ to real numbers, imposes smoothness with respect to $\mbf{x}$ on the transition probabilities $\{q_i(1|\mbf{x})\}$.

Here, we  point out that the distribution for which Proposition \ref{prop:badCorr} in Sec.\ref{sec:converse} was proved is, in fact, a special case of the generalized hidden MRF model. The main difference between the distributions in Sec.\ref{sec:converse} and Prop.~\ref{prop:hidden1DMRF} is that in the case of the former, the sum of the conditional expectations, $\sum_{i=1}^n \EX[Y_i|\mathbf{X}]$, does not satisfy the $c$-Lipschitz property. Thus, the $c$-Lipschitz property, in other words, smoothness of $\sum_{i=1}^n \EX[Y_i|\mathbf{X}]$, seems to separate good correlations from bad correlations. %However, unfortunately, we do have an intuitive physical interpretation of it. 

To derive Proposition~\ref{prop:hidden1DMRF}, we first show that if the error rate is $\epsilon$, in the above model,  the number of errors introduced between the epochs $t$ and $t+1$  is less than $n(\epsilon+\beta)$ with a probability not less than $1-\exp(-2 n \beta)$. Toward that, we derive a concentration result for non-Markov $\mbf{Y}$ by building on results from Markov concentration \cite{KontorovichR2008}. Remaining part of the proof invokes the recently derived linear distance properties of quantum Tanner codes \cite{LeverrierZ2022, LeverrierZ2023}. 

Let at every time $t$, the memory is decoded and encoded using a Tanner code with linear minimum distance $b n$ , for $b>0$ and $\beta,\epsilon\le \frac{b}{2}$. With probability at least $1-\exp(-nb)$, the number of errors introduced at $t+1$ would be at most $n b$. Thus, the Tanner code would correct the memory at time $t+1$ with probability at least $1-\exp(-nb)$, if it was correct till $t$. Thus, by the union bound, the time until memory stays accurate with probability at least $1-\frac{1}{n}$ is  lower bounded by $\exp(bn-\ln n)$. A detailed proof of Proposition~\ref{prop:hidden1DMRF} is presented in Appendix~\ref{sec:prop:hidden1DMRF}.

% In addition to the encoding and decoding using an error correcting code like Tanner code at each time $t$, suppose the correcting circuitry for the memory can shuffle the physical qubits using swap gates. Also suppose the shuffling procedure is close to uniform shuffling and there are no errors during swap gate operations. Hence, if the error rate is $\epsilon$, the error distribution across the shuffled qubits can be seen as: with probability $1-\exp(-2\beta n)$ for any $\beta>0$, at most $n(\epsilon+\beta)$ errors are distributed uniformly. Since $\beta$ can be chosen to be arbitrarily small, any code that can correct uniform errors at a rate $\epsilon$, would correct the errors in the shuffled system.  In particular, if $\mc{C}$ is a constant overhead code  with  threshold $\epsilon_{\mc{C}}$ against local stochastic errors \cite{fawzi2020constant}, then the threshold for this code against the generalized hidden Markov random field would also be $\epsilon_{\mc{C}}$. Thus, a good quantum expander or LDPC \cite{gottesman2014fault,fawzi2020constant} code can keep the memory corrected for exponential time against the non-Markov error model discussed above if the error rate is below the  noise threshold for local stochastic noise and the physical qubits can be shuffled accurately. 

As discussed above, the decay of correlation faster than $1/{\mbox{distance}}$, has generally been considered to be a requirement for mitigating correlated errors. This is echoed in the first condition of Prop.~\ref{prop:hidden1DMRF}. Interestingly, because of the better analytical tractability of the error model proposed here, we can also investigate the other side of that condition.  By a minor adaptation in the proof of Prop.~\ref{prop:hidden1DMRF}, it follows that if the correlation between errors at two different locations decays as $1/\text{distance}^{1-\frac{\eta}{2}}$, Tanner codes can achieve a retention time that is exponential in $n^{1-\eta}$. To be more specific, if the sum of the absolute pairwise correlations, $\max_{1\le i \le n} \sum_{k=i}^n \prod_{j=i}^k \theta_j$ (Prop.~\ref{prop:hidden1DMRF}), is $g(n)$, then the retention time is lower bounded by the exponential of a constant times $\frac{n}{(g(n))^2}$. Thus, a correlation decay faster than $1/\text{distance}$ for $1$D systems is not necessary for a super-polynomial retention time of the form $\exp(n^{1-\eta})$ for some $\eta \in (0,1)$.

% add to a value that scales with $n$, say $g(n)$, a Tanner code based fault-tolerance scheme can still guarantee a retention time which is expoential of a  constant times $\frac{n}{(g(n))^2}$.  In particular, for the error model in Prop.~\ref{prop:hidden1DMRF}, even for a correlation decay $1/\text{distance}^{1-\frac{\eta}{2}}$, Tanner codes can achieve a retention time  which is exponential in $n^{1-\eta}$.

\subsection*{\bf Conclusion}
In summary, we proposed a hidden random field model for long-range correlated errors, which is a broad generalization of the local stochastic noise model in \cite{AharonovB2008}.  The proposed model provably captures correlations outside the scope of the joint Hamiltonian model with pairwise terms and,   arguably, has a strictly larger expanse. This model is also analytically tractable due to the rich literature on probability and statistical physics around it \cite{KunschGK1995,Chatterjee2018,CamiaJN2018}. Building on this literature and the minimum distance properties of Tanner codes, we proved that Tanner codes can keep a memory corrected for exponential time against a broad class of non-Markov and possibly non-stationary long-range correlated errors. 

Based on results from pairwise Hamiltonian models, it is generally believed that for $1$D qubit arrangements, a correlation decay faster than $1/\text{distance}$  is necessary and sufficient for arbitrarily long fault-tolerance. However, as discussed above, for the model in Sec.\ref{sec:tanner} whose conditional distributions satisfy the $c$-Lipschitz property, faster than $1/\text{distance}$ correlation decay was not necessary for having super-polynomial (in $n$) retention time. On the other hand, for the error distribution in Sec.\ref{sec:converse}, which has a correlation decay faster than $1/\text{distance}$, but does not satisfy $c$-Lipschitz property, no fault-tolerant scheme can have high retention time. Thus, the $c$-Lipschitz property, which is crucial for measure  concentration of non-i.i.d. distributions \cite{KontorovichR2008,BoucheronLM2013} seems to be the differentiator, not the rate of correlation decay. An implication is that the concentration of error distributions seems to be a more accurate criterion for identifying the class of correctable correlated errors in quantum memory than the rate of correlation decay.

\subsection*{{\bf Acknowledgments}}
The authors thank Arul Lakshminarayan for helpful comments on the draft. The authors thank ANRF, India (formerly SERB, India) for support through CRG/2023/005345 and the Ministry of Education, India.

% %%%%%%%%%%%%%%%%%%%%%%%%%%%%%%%%%%%%%%%%%%%%%%%%%%%%%%%%%%%%%%%%%%%%%%%%%%%%%%%%%
% \vspace{1 in}

% {\color{green} \bf Previous write up below this}

% \vspace{1 in}
% %%%%%%%%%%%%%%%%%%%%%%%%%%%%%%%%%%%%%%%%%%%%%%%%%%%%%%%%%%%%%%%%%%%%%%%%%%%%%%%%%

%\bibliographystyle{IEEEtran}
\bibliography{bibfile}

\appendix

\section{Proof of Proposition~\ref{prop:badCorr}}
\label{sec:prop:badCorr}
Consider the following non-Markovian noise model that was introduced in Proposition~\ref{prop:badCorr}. The error variable at $i$ is denoted by $Y_i$. Let $X_i$ for $i=1,2,\hdots n$ be an i.i.d Bernoulli random variable with $P(X_i=1)=\epsilon$ and $B_n:=n\epsilon + \sqrt{n}C_n$. The error variable $Y_i=X_i$ for all $i=1,2,\hdots n$ if $\sum\limits_{i=1}^{n} X_i \leq B_n$; and $Y_i=1$ for all $i=1,2,\hdots n$ if $\sum\limits_{i=1}^{n} X_i > B_n$.

Since partial trace and expectation interchange,  $||\mc{N}_{iJ}-\mc{N}_i \otimes \mc{N}_j||_{\diamond}$ is covariance of $Y_i$ and $Y_j$ times $||\mc{E}\otimes \mc{E}||_{\diamond}$. Thus, an upper bound on the covariance between $Y_i$ and $Y_j$ also implies the same  upper bound on $||\mc{N}_{iJ}-\mc{N}_i \otimes \mc{N}_j||_{\diamond}$. 

Therefore, we now investigate the covariance for this noise model and the memory length it permits.

\begin{theorem}
\label{thm:correlation}
    The covariance between any two error variables $Y_i$ and $Y_j$ for the particular error distribution is less than $c_3\exp(-c_6 C_n^2)$ for some $c_3 >0$ and $c_6>0$.
\end{theorem}

 To prove Theorem~\ref{thm:correlation}, we find an expression for the covariance between any $Y_i$ and $Y_j$ in terms of the conditional covariance given the event $A:=\{\sum\limits_{i=1}^{n} X_i > B_n\}$. Then, the law of total covariance allows us to combine the bounds on conditional covariance to derive the bound for the unconditional covariance between $Y_i$ and $Y_j$.
 
 Let $\I_A$ denote the indicator of event $A$, i.e., the random variable $\I_A:= 1$ if event $A$ occurs and $0$ otherwise. The law of total covariance is applied here as follows: 
 \begin{align}
 \label{eq:total_covariance}
     \cov(Y_i,Y_j) &= \EX[\cov(Y_i,Y_j \mid \I_A)] %\nonumber \\
        + \cov(\EX[Y_i \mid \I_A], \EX[Y_j \mid \I_A]) 
 \end{align}

The following Lemmas (\ref{lem:cov_given_Acomplement} and \ref{lem:cov_expectation}) will investigate bounds on the various terms in eq~\eqref{eq:total_covariance}.
\begin{lemma}
\label{lem:cov_given_Acomplement}
     The conditional covariance of $Y_i$ and $Y_j$ for any $i \neq j$ given $\I_A=0$ decays exponentially in $C_n^2$, i.e., $\cov(Y_i,Y_j \mid \I_A=0) \leq c_4\exp(-c_6C_n^2)  $ for some $c_4$, $c_6>0$.
\end{lemma}

%Proof of Lemma~\ref{lem:cov_given_Acomplement}
\begin{proof}
    Given $\I_A=0$, $Y_i=X_i$ and $Y_j=X_j$. Substituting this in the definition of conditional covariance, we have $\cov(Y_i,Y_j \mid \I_A=0) = \EX[(X_i - \EX[X_i \mid \I_A=0]) (X_j-\EX[X_j \mid \I_A=0]) ]$. 
    
    Let $\EX_P[\cdot]$ denote expectation over the joint distribution $P:=p_{X_i,X_j \mid \I_A=0}(\cdot)$. Then $\cov(Y_i,Y_j \mid \I_A=0) = \EX_P[f(X_i,X_j)]$. 
    
    Let $Q$ denote the product distribution $p_{X_i \mid \I_A=0}\cdot p_{X_j \mid \I_A=0}$. We can see that $\EX_Q[f(X_i,X_j)]=0$. Therefore, we can write $\cov(Y_i,Y_j \mid \I_A=0)= \EX_P[f(X_i,X_j)] - \EX_Q[f(X_i,X_j)]$. From \cite{BoucheronLM2013}, $\EX_P[f(X_i,X_j)] - \EX_Q[f(X_i,X_j)] \leq \|P - Q\|_1$. 

    Now consider $\|P - Q\|_1$. By definition, 
    % \begin{align*}
    %     \|P - Q\|_1 &= \sum \limits_{\{x_i,x_j = 0,1\}} |\PR(X_i=x_i, X_j=x_j \mid \I_A=0) \nonumber \\ & ~~ - \PR(X_i=x_i\mid \I_A=0)\PR( X_j=x_j \mid \I_A=0)| \nonumber \\
    % & = \sum \limits_{\{x_i,x_j = 0,1\}} \Bigg| \frac{\PR(X_i=x_i, X_j=x_j, A^c)}{\PR(A^c)} \nonumber \\ 
    %     & ~~ - \frac{\PR(X_i=x_i, A^c)}{\PR(A^c)}\cdot \frac{\PR(X_j=x_j, A^c)}{\PR(A^c)}\Bigg|
    % \end{align*}

    By rewriting the joint and conditional distributions in the expression differently, $\|P - Q\|_1$ becomes
    % \begin{align*}
    %    & 
    %     \frac{\PR(X_i=x_i \mid A^c)\PR(X_i=x_i)\PR(X_j=x_j)}{\PR(X_i=x_i, A^c)} \\
    %     & \times \sum | \PR(A^c \mid X_i=x_i,X_j=x_j) -\PR(A^c|X_i=x_i)\PR(A^c|X_j=x_j) |
    % \end{align*}

    \begin{align*}
       & 
        \frac{\PR(X_i=x_i)\PR(X_j=x_j)}{\PR(A^c)} \big( 
        \sum | \PR(A^c \mid X_i=x_i,X_j=x_j)  \nonumber \\
        &  -\PR(A^c|X_i=x_i)\PR(A^c|X_j=x_j) | \big)
    \end{align*}

    Consider $\PR(A^c \mid X_i=0,X_j=0)$:
    \begin{align*}
        \PR(A^c \mid X_i=0,X_j=0) = \PR(\sum \limits_{k \neq i,j} X_k \leq B_n)
    \end{align*}
    For simplicity, we denote the probability of the sum of any $(n-1)$ $X_i's$ being greater than $x$ by $\delta_{n-1}$, i.e., probability of the event event $\sum \limits_{k \neq i} X_k > x$ by $\delta_{n-1}$. Similarly use the notation $\delta_{n-2}$ for sum of $(n-2)$ $x_i
    s$.
    % By central limit theorem (CLT), $\frac{\sum \limits_{k \neq i,j} X_k - (n-2)\epsilon}{\sqrt{(n-2)}}$ is normally distributed with mean $0$ and variance $=\vari(X_i)$. Therefore by \cite[Thm 1]{Chang}, $\PR(A \mid X_i=0,X_j=0) := \delta_{n-2} \leq c_5\exp(-c_6(2\epsilon+C_n^2))$, which is less than or equal to $c_5\exp(-c_6C_n^2)$. 
    By the use of Azuma-Hoeffding inequality for $\{0,1\}$-valued i.i.d random variables \cite{BoucheronLM2013}, it can be seen that
    $\delta_{n-1}$ and $\delta_{n-2}$ are upper boubded by $c_5\exp(-c_6 C_n^2)$, for some $c_5$, $c_6$ $>0$.

    Thus, 
    \begin{align*}
        &\left| \PR(A^c \mid X_i=0,X_j=0) - \PR(A^c|X_i=0)\PR(A^c|X_j=0) \right| \\
        &= |1-\delta_{n-2} - (1-\delta_{n-1})^2| \\
        &=|\delta_{n-2} - \delta_{n-1}^2 + 2\delta_{n-1}|\\
        &\leq 3 c_5\exp(-c_6C_n^2)
    \end{align*}

    Similarly, we can get a similar bound for  other $x_i$ and $x_j$ values. Note that for any non-decreasing $C_n$, $\frac{\PR(X_i=x_i)\PR(X_j=x_j)}{\PR(A^c)} \le C \epsilon^2$ for any $x_i$ and $x_j$, for some $C>0$. Combining these we get
    \begin{align*}
        \|P - Q\|_1 \leq c_4\exp(-c_6~C_n^2),
    \end{align*}
    for some $c_4>0$.
\end{proof}

\begin{lemma}
\label{lem:cov_expectation}
    $\cov(\EX[Y_i \mid \I_A], \EX[Y_j \mid \I_A]) \leq c'_5\exp(-c'_6 C_n^2)$ for some $c'_5$ and $c'_6>0$.
\end{lemma}
\begin{proof}
    Let $Z_i:=\EX[Y_i \mid \I_A]$ and $Z_j:=\EX[Y_j \mid \I_A]$. $Z_i$ and $Z_j$ are functions of $\I_A$ denoted by $f_i(\I_A)$ and $f_j(\I_A)$. Since $\I_A$ is a $\text{Bernoulli}(p)$ random variable, the variance of  $f_i(\I_A)$ is $p(f_i(1))^2 + (1-p)(f_i(0))^2 - (f_i(1)p + f_i(0)(1-p))^2 = p(1-p)(f_i(1)-f_i(0))^2$. As $f_i(\I_A) \in [0,1]$, we have $\vari(Z_i) \leq \vari(I_A)$. Therefore, $\cov(Z_i,Z_j) \leq \sqrt{\vari(Z_i)\vari(Z_j)} \leq \vari(\I_A)$. Variance of $\I_A$ is $P(A)(1-P(A)) \leq P(A) \leq c'_5\exp(-c'_6 C_n^2)$
\end{proof}

%\subsubsection*{Proof of Theorem~\ref{thm:correlation}}
\begin{proof}[Proof of Theorem~\ref{thm:correlation}]
    By the law of total covariance,
    \begin{align*}
    %\label{eq:total_covariance}
        \cov(Y_i,Y_j) &= \EX[\cov(Y_i,Y_j \mid \I_A)] \\
        &+ \cov(\EX[Y_i \mid \I_A], \EX[Y_j \mid \I_A])         
    \end{align*}
    The first term in the above expression is $\EX[\cov(Y_i,Y_j \mid \I_A)]= \PR(A)\cov(Y_i,Y_j \mid \I_A=1) + \PR(A^c)\cov(Y_i,Y_j \mid \I_A=0)$. Given $\I_A=1$, $Y_i=1$ for all $i=1,2,\hdots n$. Therefore, $\cov(Y_i,Y_j \mid \I_A=1) = \EX[Y_i Y_j \mid \I_A=1] - \EX[Y_i \mid \I_A=1]\EX[Y_j \mid \I_A=1] = 1 -1\cdot 1 = 0$. 

    By Lemma~\ref{lem:cov_given_Acomplement}, we have $\cov(Y_i,Y_j \mid \I_A=0)\leq c_4\exp(-c_6C_n^2)$. Therefore the first term, i.e. $\EX[\cov(Y_i,Y_j \mid \I_A)] \leq \PR(A^c)\cdot c_4\exp(-c_6C_n^2)$. 

    By Lemma~\ref{lem:cov_expectation}, the second term is  $\cov(\EX[Y_i \mid \I_A], \EX[Y_j \mid \I_A]) \leq c'_5\exp(-c'_6 C_n^2)$. Thus, we get the bound in Theorem \ref{thm:correlation}.
\end{proof}

\begin{lemma}
\label{lem:length}
    Retention time of the memory is no more than $\frac{1}{\PR(A)}$ for any tolerable distance (from the initial state) $\le 1-\frac{1}{e}$.
\end{lemma}
\begin{proof}
    All information in the qubits is lost when the error state $Y_{i,t}=1$  for all $i$ occurs at some epoch $t$. The probability that the error state occurs at a particular error correction epoch is $\PR(A)$ and is independent of other epochs since $\mbf{X}_t$ are chosen to be i.i.d. in time. 
    
    The probability of the memory is fully lost by epoch $l$ is $1-(1-\PR(A))^l$. This is lower bounded by $1-\exp(-\PR(A) l)$ since $1-x \le e^{-x}$ for all $x \ge 0$. Thus, for any $l \ge \frac{1}{\PR(A)}$ with probability at least $1-\frac{1}{e}$, the memory is lost.  Thus, the retention time of the memory is less than $\frac{1}{P(A)}$ for any accuracy requirement $\le 1-\frac{1}{e}$.
\end{proof}

By concentration of i.i.d. Bernoulli random variables \cite{Krishnapur2016notes},  $P(A)=P(Z \in [\frac{C_n}{2},\infty))$, where $Z$ is a Gaussian RV with zero mean and variance $\epsilon(1-\epsilon)$. By \cite[Thm 2]{Chang}, $P(A)$ is lower bounded by $c_7 \exp(-c_8C_n^2)$, for some $c_7$ and $c_8$ $>0$. Hence, by Lemma \ref{lem:length}, the retention time is upper bounded by $(1/c_7) \exp(c_8C_n^2)$.

% By central limit theorem, $P(A)=\text{constant}\cdot\text{erfc}(C_n^2)$, where $\text{erfc}(\cdot)$ is the complementary error function. By \cite[Thm 2]{Chang}, $P(A)$ is lower bounded by $c_7 \exp(-c_8C_n^2)$, which means that the retention time is upper bounded by  $(1/c_7) \exp(c_8C_n^2)$. 

Finally, choosing $C_n=a \sqrt{\ln n}$ for a suitable $a>0$ gives the result in the proposition.

\section{Proof of Proposition \ref{prop:hidden1DMRF}}
\label{sec:prop:hidden1DMRF}
This proposition is proved using the following lemma.
\begin{lemma}
\label{lem:finalConc}
If the error rate is upper bounded by $\epsilon$, then
$\PR(\sum_{i=1}^n Y_i > n (\epsilon+\delta)) \le c_1 \exp(-c_2 \delta^2 n) $ for some positive constants $c_1$, $c_2$ independent of $\epsilon$, $\delta$ and $n$.
\end{lemma}
Proof of this lemma is presented later. First, we prove Proposition \ref{prop:hidden1DMRF} using this lemma.

Suppose the Tanner code has a minimum distance $c n$ for some $0<c<1$. Lemma \ref{lem:finalConc} implies that if $\epsilon$ is no more than $\frac{c}{2}$ and $\delta$ is chosen to be $\frac{c}{2}$, then we obtain that the probability that the number of errors at an epoch is more than the minimum distance of the code $\le c_1 \exp(-c_2 \frac{c^2}{2} n)$. Since errors across epochs are independent, the probability that there will be at least one epoch till $T$ when the number of errors are more than the minimum distance is bounded by $r=T c_1 \exp(-c_2 \frac{c^2}{2} n)$, thanks to the union bound. So, with probability at least $r$, at each epoch till $T$, the decoded state will be exactly equal to the initial state. Thus, the distance between any decoded state till $T$ and the original state is upper bounded by $r$. This gives an exponential lower bound on $T$.

Errors in the correcting circuit can be guarded against in a way similar to \cite{BravyiCG2024high}, if these errors are i.i.d. and the rate is sufficiently small.

\begin{proof}[Proof of Lemma \ref{lem:finalConc}]
\begin{align}
& \PR\left(\sum_{i=1}^n Y_i > n \epsilon + n \delta \right) \nonumber = \PR\left(\sum_{i=1}^n Y_i - \sum_{i=1}^n \EX[Y_i] > n \delta \right) \nonumber \\
& \le \PR\left(|\sum_{i=1}^n Y_i - \sum_{i=1}^n \EX[Y_i]| > n \delta \right) \nonumber \\
& = \sum_{\mbf{x}} \mbf{p}(\mbf{x}) \PR\left(|\sum_{i=1}^n Y_i - \sum_{i=1}^n \EX[Y_i]| > n \delta | \mbf{X}=\mbf{x}\right). \nonumber
\end{align}
Now,
\begin{align}
&~~~~~\PR\left(|\sum_{i=1}^n Y_i - \sum_{i=1}^n \EX[Y_i]| > n \delta | \mbf{X}=\mbf{x}\right) \nonumber \\
& = \PR\left(|\sum_{i=1}^n Y_i - \sum_{i=1}^n \EX[Y_i|\mbf{X}=\mbf{x}] + \sum_{i=1}^n \EX[Y_i|\mbf{X}=\mbf{x}] \right. \nonumber \\
~~~~~~ & \left. - \sum_{i=1}^n \EX[Y_i]| > n \delta | \mbf{X}=\mbf{x}\right) \nonumber \\
% \end{align}
% \begin{align}
& \le \PR\left(|\sum_{i=1}^n Y_i - \sum_{i=1}^n \EX[Y_i|\mbf{X}=\mbf{x}]|> n \frac{\delta}{2} | \mbf{X}=\mbf{x}\right) \nonumber \nonumber \\
&~~+  \PR\left(|\sum_{i=1}^n \EX[Y_i|\mbf{X}=\mbf{x}] - \sum_{i=1}^n \EX[Y_i]| > n \frac{\delta}{2} | \mbf{X}=\mbf{x}\right).
\end{align}
Last inequality is obtained by applying the following facts in their respective order: (i) if for $a, b, c>0$, $a + b > c$, then at least one of $a$ and $b$ must be greater than $\frac{c}{2}$  (ii) monotonicity of probability,  and (iii) union bound.

The following two lemmas are useful in the rest of the proof.

\begin{lemma}
\label{lem:condIIDconc}
For any given realization  $\mbf{X}=\mbf{x}$, $\PR(|\sum_{i=1}^n (Y_i - \EX[Y_i|\mbf{X}=\mbf{x}])| \ge \beta n |\mbf{X}=\mbf{x}) \le 2 \exp(-\beta^2 n)$ for any $\beta>0$.
\end{lemma}

\begin{lemma}
\label{lem:MarkovConcFinal}
For any $\beta>0$, under conditions  in Proposition \ref{prop:hidden1DMRF},  $\PR(|\sum_{i=1}^n (\EX[Y_i|\mbf{X}=\mbf{x}]-\EX[Y_i])| \ge \beta n) \le 2 \exp(-c_3 \beta^2 n)$ for some $c_3>0$, which is independent of $n$.
\end{lemma}

Clearly, $ \PR\left(|\sum_{i=1}^n \EX[Y_i|\mbf{X}=\mbf{x}] - \right.$ $\left.\sum_{i=1}^n \EX[Y_i]| > n \frac{\delta}{2} | \mbf{X}=\mbf{x}\right)$ is either $1$ or $0$ depending on whether the following condition is satisfied for $\mbf{x}$ or not: $|\sum_{i=1}^n \EX[Y_i|\mbf{X}=\mbf{x}] - \sum_{i=1}^n \EX[Y_i]| > n \frac{\delta}{2}$.

Thus,  $\sum_{\mbf{x}} \mbf{p}(\mbf{x})$ $ \PR\left(|\sum_{i=1}^n \EX[Y_i|\mbf{X}=\mbf{x}] - \right.$ $\left.\sum_{i=1}^n \EX[Y_i]| > n \frac{\delta}{2} | \mbf{X}=\mbf{x}\right)$ is equal to $\PR(|\sum_{i=1}^n (\EX[Y_i|X_i]-\EX[Y_i])| \ge \beta n)$. This has an exponentially vanishing upper bound from Lemma~\ref{lem:MarkovConcFinal}.

On the other hand, $\PR\left(|\sum_{i=1}^n Y_i - \sum_{i=1}^n \EX[Y_i|\mbf{X}=\mbf{x}]|> n \frac{\delta}{2} | \mbf{X}=\mbf{x}\right)$ is uniformly upper bounded for all $\mbf{x}$ in Lemma~\ref{lem:condIIDconc}. Hence, the summation with respect to $\sum_{\mbf{x}} \mbf{p}(\mbf{x})$ remains the same.

Thus, we get the final bound in Lemma \ref{lem:finalConc}.
\end{proof}

\begin{proof}[Proof of Lemma~\ref{lem:condIIDconc}]
By the definition of the proposed $1$D generalized hidden MRF model, $Y_i$ is conditionally independent of $\{Y_j: j \neq i\}$ given $\mbf{X}=\mbf{x}$. Hence, this lemma follows directly from the Hoeffding inequality \cite{BoucheronLM2013} for independent (but not necessarily identically distributed) random variables by taking $[0,1]$ as the range of these random variables.
\end{proof}

\begin{proof}[Proof of Lemma~\ref{lem:MarkovConcFinal}]

Note that $\sum_i \EX[Y_i|\mbf{X}=\mbf{x}]$ is a $c$-Lipschitz function of $\mbf{x}$. Also, $\mbf{X}$ is a $1$D Markov process. Hence, we use the concentration result for inhomogeneous Markov chains from \cite{KontorovichR2008}, in particular, Theorem 1.2. Note that by the condition in Proposition~\ref{prop:hidden1DMRF}, the quantity $M_n$ in \cite[Thm~1.2]{KontorovichR2008} is a constant independent of $n$. Thus, by plugging in all these terms in \cite[Thm~1.2]{KontorovichR2008} we get the desired
bound in Lemma~\ref{lem:MarkovConcFinal}.
\end{proof}

%Thus, the inequality we need to prove is 

% \noindent $\PR(|\sum_{i=1}^n (f_i(X_i)-\EX[f_i(X_i)])| \ge \beta n)$ $ \le 2 \exp(-c_3 \beta^2 n).$

%\noindent Since error rate is $\epsilon$, $\sum_{i=1}^n \EX[f(X_i)]=\sum_{i=1}^n \EX[Y_i]=n\epsilon$. %Next, we apply Theorem 3 from \cite{ChungLLM2012} with the following parameters: $T=\tau$, $\mu=\epsilon$, $\delta=\frac{\beta}{\epsilon}$. Here the parameters on the left of the equality signs are from Theorem 3 from \cite{ChungLLM2012}. This proves Lemma~\ref{lem:MarkovConcFinal} under condition 1 of Proposition~\ref{prop:hidden1DMRF}.

% Next, we prove Lemma~\ref{lem:MarkovConcFinal} under condition 2 of Proposition~\ref{prop:hidden1DMRF}. Towards this, 

% First note that for any $q_i$, the joint function $(f_1, f_2, \ldots f_n)$ from $\mc{X}^n$ to $[0,1]^n$ is equivalent to $\varphi$ in \cite[Thm~1.2]{KontorovichR2008}. Clearly this function is $1$-Lipschitz, i.e., $c$ in \cite[Thm~1.2]{KontorovichR2008} is $1$ in our case. Also, by condition 2, the term $M_n$ in \cite[Thm~1.2]{KontorovichR2008} is upper bounded by a constant independent of $n$. Thus, by plugging in all these terms we get the desired bound in Lemma~\ref{lem:MarkovConcFinal} under condition 2 of Proposition~\ref{prop:hidden1DMRF}.

% \bibliography{bibfile}

% \end{document}

\end{document}